\documentclass[11pt]{article}

\usepackage[T1]{fontenc}
\usepackage[utf8]{inputenc}

\usepackage[margin = 0.9in]{geometry}
\usepackage{booktabs} 
\usepackage{enumitem}

\usepackage{latexsym,amsmath,amsfonts,amssymb,stmaryrd,mathtools}
\usepackage{amsthm}
\usepackage{xcolor}
\usepackage{algorithm, algorithmicx}
\usepackage[noend]{algpseudocode}
\usepackage{graphicx}
\usepackage{hyperref}
\usepackage{wrapfig}
\usepackage{listings, fancyvrb, multicol}
\usepackage{multirow}
\usepackage{times}
\usepackage{comment}
\usepackage{authblk}
\usepackage{xspace}
\usepackage{pifont}
\usepackage{parcolumns}
\usepackage{textcomp}

\definecolor{ao}{rgb}{0.0, 0.5, 0.0}

\newtheorem{theorem}{Theorem}[section]
\newtheorem{lemma}[theorem]{Lemma}

\newtheorem{definition}[theorem]{Definition}

\newcommand{\hide}[1]{}

\makeatletter
\lst@Key{countblanklines}{true}[t]%
{\lstKV@SetIf{#1}\lst@ifcountblanklines}

\lst@AddToHook{OnEmptyLine}{%
	\lst@ifnumberblanklines\else%
	\lst@ifcountblanklines\else%
	\advance\c@lstnumber-\@ne\relax%
	\fi%
	\fi}
\makeatother

\usepackage{environ}

\newif\ifappendix
\NewEnviron{maybeappendix}[1]
{\ifappendix
	\expandafter\global\expandafter\let\csname putmaybeappendix#1\endcsname\BODY%
	\else
	\expandafter\newcommand\csname putmaybeappendix#1\endcsname{}\BODY%
	\fi}
\newcommand{\putmaybeappendix}[1]{\csname putmaybeappendix#1\endcsname}

\appendixtrue

\title{Space Complexity of Implementing Large Shared Registers}

\author{Yuanhao Wei}
\affil{Carnegie Mellon University}
\affil{yuanhao1@cs.cmu.edu}

\begin{document}
	\maketitle
	
	\begin{abstract}
	    We prove two new space lower bounds for the problem of implementing a large shared register using smaller physical shared registers. We focus on the case where both the implemented and physical registers are single-writer, which means they can be accessed concurrently by multiple readers but only by a single writer. To strengthen our lower bounds, we let the physical registers be atomic and we only require the implemented register to be regular. Furthermore, the lower bounds hold for obstruction-free implementations, which means they also hold for lock-free and wait-free implementations.
	    
	    If $m$ is the number values representable by the large register and $b$ is the number of values representable by each physical register, our first lower bound says that any obstruction-free implementation that has an invisible reader requires at least $\lceil \frac{m-1}{b-1} \rceil$ physical registers. A reader is considered invisible if it never writes to shared registers. This lower bound is tight for the invisible reader case. We also prove a \hbox{$\lceil \min(\frac{m-1}{b-1}, r+\frac{\log{m}}{\log{b}}) \rceil$} space lower bound for the general case, which covers both visible and invisible readers. In this bound, $r$ represents the number of readers.
	\end{abstract}
	
	\nocite{*}
	\section{Introduction}
\label{sec:intro}

In most shared memory multi-processor systems, processes communicate with each other by reading from and writing to shared registers. Modern systems typically allow you to atomically read and write some constant number of bits.
To read and write larger amounts of data, you would have to implement a larger register using the smaller physical ones provided by the system. This paper studies the space complexity required by such implementations. We define the space complexity of an implementation to be the number of physical registers it uses. The step complexity of an operation is defined to be the worst case number of steps needed to complete the operation.

This problem varies in several dimensions. There are three
common correctness conditions for shared registers, \emph{safe}, \emph{regular}, and \emph{atomic}, which were introduced by Lamport back in 1986 \cite{lamport1986interprocess}. Atomicity is the strongest of the three conditions and safety is the weakest. All atomic registers are regular and all regular registers are safe. In this paper we only consider regular and atomic registers. Shared registers can also differ in the number of readers and the number of writers allowed to access the register concurrently. In this paper we only consider \emph{single-writer} (SW) registers, which are registers that can be accessed concurrently by multiple readers but only a single writer. Finally, we restrict our attention to non-blocking implementations. This excludes the use of locks and other blocking techniques. There are three common non-blocking progress guarantees that appear in the literature: \emph{obstruction-freedom}, \emph{lock-freedom}, and \emph{wait-freedom}. Obstruction-freedom is the weakest natural non-blocking guarantee and it includes all lock-free and wait-free algorithms. Wait-freedom is the strongest guarantee and it ensures that every process makes progress regardless of how the processes are scheduled. The terms \emph{obstruction-freedom}, \emph{regular}, and \emph{atomic} are defined formally in Section \ref{sec:model}.

Table \ref{tbl:impl} lists some previous implementations of an $m$-value SW register from $b$-value SW registers. The number of readers is represented by $r$. A reader is considered to be invisible if it never writes to shared registers. The \textquotesingle{}Invisible\textquotesingle{} column contains a \textquotesingle{}yes\textquotesingle{} whenever the implementation has at least one invisible reader. All implementations listed in the table are wait-free. The register type of an implementation is atomic, if it implements a large atomic register using smaller atomic registers. Similarly we say it's regular if it is a regular from regular implementation. Some papers \cite{Larsson} assume additional atomic primitives like, swap, fetch-and-add, and compare-and-swap, but we focus on implementations without any additional primitives.

{\renewcommand{\arraystretch}{1.5}
\begin{table}[H]
\small
\label{tbl:impl}
	\centering
    \begin{tabular}{ | l | c | c | c | c | c | c | }
    \hline
    \textbf{Prior Work} & \textbf{Register Type} & \textbf{Invisible?} & \textbf{Space} & & \textbf{Read} & \textbf{Write}  \\
    \hline
     Peterson \cite{Peterson} & atomic & no & $\Theta(r \frac{\log m}{\log b})$ & &$\Theta(\frac{\log m}{\log b})$ & $\Theta(r \frac{\log m}{\log b})$ \\
    \hline
    Chaudhuri and Welch \cite{chaudhuri1994bounds} & regular & yes & $\Theta(\frac{m}{b})$ & & $\Theta(\frac{\log m}{\log b})$ & $\Theta(\frac{\log m}{\log b})$  \\
    \hline
    Vidyasankar \cite{vidyasankar1988converting} & atomic & yes & $\Theta(\frac{m}{\log b})$ & & $\Theta(\frac{m}{\log b})$ & $\Theta(\frac{m}{\log b})$  \\
    \hline
    Chaudhuri, Kosa and Welch \cite{chaudhuri2000one} & regular, atomic & yes & $\Theta(m^2)$ & & $\Theta(m^2)$ & 1 \\
    \hline
    Chen and Wei \cite{chen2017step} & atomic & yes & $\Theta(\frac{m^2}{b^2})$ & & $\Theta(\frac{\log m}{\log b})$ & $\Theta(\frac{\log m}{\log b})$  \\
    \hline
     Chen and Wei \cite{chen2017step} & atomic & no & $O(r\frac{\log m}{\log b})$ & & $\Theta(\frac{\log m}{\log b})$ & $\Theta(\frac{\log m}{\log b})$ \\
    \hline
    \end{tabular}
    \caption{Wait-free $m$-value SW register implementations from $b$-value SW registers}
    \label{complexity-table}
\end{table}}

Notice that all implementations with an invisible reader use at least $\Theta(\frac{m}{b})$ space and all implementations with visible readers use at least $\Theta(r)$ space. This paper helps explain the high space usage of these implementations by showing that any obstruction-free, regular from atomic implementation requires at least $\lceil \frac{m-1}{b-1} \rceil$ space in the invisible reader case and $\lceil \min(\frac{m-1}{b-1}, r+\frac{\log{m}}{\log{b}}) \rceil$ space in the general case. Chaudhuri and Welch's implementation \cite{chaudhuri1994bounds} shows that our lower bound is asymptotically tight for the invisible reader case. Their implementation was first introduced for the $b=2$ case. Later, Chen and Wei \cite{chen2017step} show how it can be generalized for any $b \geq 2$. When $m$ is a power of $b$, the number of registers used by the implementation is $\frac{m-1}{b-1}$, which matches our lower bound exactly.

There are some previous space lower bounds for this problem. Chaudhuri and Welch prove multiple lower bounds in \cite{chaudhuri1994bounds}. The one that is most relevant to this paper says that any regular from regular implementation where $b=2$ requires at least $\lceil \max(\log m + 1, 2\log m - \log \log m - 2) \rceil$ space. Chaudhuri, Kosa and Welch \cite{chaudhuri2000one} prove that any regular from regular implementation where $b=2$ and the writer only performs a single operation requires $\Omega(m^2)$ space. This shows that their one-write algorithm is space optimal. They also prove a space lower bound of $2m-1-\lceil \log m\rceil$ for a slightly more general case. Berger, Keidar and Spiegelman \cite{berger2018integrated} consider a class of algorithms where each read operation has to see at least $\tau \geq 2$ values written by the same write operation before the read operation is allowed to return. They show that in this setting, any wait-free, regular from atomic implementation requires $\tau m$ space for the invisible reader case and $\tau + (\tau-1)\min(m-1, r)$ space for the general case. This lower bound helps explain the space complexity of Peterson's \cite{Peterson} as well as Chen and Wei's \cite{chen2017step} implementation because $\tau = \frac{\log m}{\log b}$ in both implementations. However, the lower bound does not apply to any of the invisible reader algorithms from Table \ref{tbl:impl} because $\tau$ equals $0$ or $1$ in all those algorithms.

We define some important terms in Section \ref{sec:model} and we prove both our lower bounds in Section \ref{sec:lb}.


	\section{Model}
\label{sec:model}

A {\em single-writer (SW) register} $R$ is a shared register where only one process can perform write operations and any number of processes can perform read operations. We say that a process owns $R$ if it can write to $R$. We will work in the standard asynchronous shared memory model \cite{attiya2004distributed} with $r$ readers and one writer, which communicate through shared physical registers. Processes may fail by crashing. 

In our model, an \emph{execution} is an alternating sequence of \emph{configurations} and \emph{steps} $C_0$, $e_1$, $C_1$, $e_2$, $C_2$, $\dots$, where $C_0$ is an \emph{initial configuration}. Each step is either a read or write of a physical register. Configuration $C_i$ consists of the state of every register and every process after the step $e_i$ is applied to  configuration $C_{i-1}$. 

A register is \emph{atomic} if its operations are linearizable \cite{HW}. A register is \emph{regular} if the value returned by each read is either the value written by the last write operation completed before the first step of the read or the value written by a write operation concurrent with the read operation. Note that every atomic register is also regular.


The rest of this paper will focus on the \emph{obstruction-free} progress guarantee which says that if at any point in the execution, an operation is allowed to run in isolation (with all other processes suspended), then it will terminate in a finite number of steps. Any \emph{wait-free} and \emph{lock-free} algorithm is also obstruction-free.
    
\section{Space Lower Bounds}
\label{sec:lb}
    This section proves two new lower bounds on the number of atomic registers needed to implement a large regular register. The term \textquotesingle{}implementation\textquotesingle{} will frequently be used as a shorthand which means \textquotesingle{}obstruction-free implementation of a regular SW register from smaller atomic SW registers\textquotesingle{}. The first lower bound, Theorem \ref{thm:lb_inv}, applies to all implementations with an invisible reader. This lower bound can be used to easily prove a more general lower bound that holds for the visible reader case as well. This is done in Theorem \ref{thm:lb_vis}. 
    
    Throughout the proofs, there are three important algorithm parameters that come up repeatedly: $m$, the number of values that can be represented by the simulated register, $n$, the number of physical registers in the implementation, and finally, $S$, the total fanout of the implementation. If each of the physical registers can represent $b$ values, then the total fanout is simply $nb$. However, it greatly simplifies the proofs to consider implementations that use physical registers of different sizes.
    Below is the definition of \textquotesingle{}total fanout\textquotesingle{} for this more general setting.
    
    \begin{definition}
        Let $A$ be an implementation and let $b_i$ be the number values that can be represented by the $i^\text{th}$ physical register. The \textbf{total fanout} of $A$ is defined to be the sum of all the $b_i$'s.
    \end{definition}
    
    Here is an overview of the proofs in this section. The first proof is for the main technical lemma, Lemma \ref{lem:lb_main}, which says if there exists an $m$-value register implementation with $S$ fanout and an invisible reader, then there exists an $(m-1)$-value register implementation with $S-1$ fanout and an invisible reader. Once this lemma is established, the rest of the proof is straight forward. Lemma \ref{lem:inv} uses Lemma \ref{lem:lb_main} inductively to argue that $S$ must be large when $m$ is large. The invisible reader lower bound, Theorem \ref{thm:lb_inv}, is basically a special case of Lemma \ref{lem:inv} where all the physical registers have the same size. And finally a short proof of the general lower bound, Theorem \ref{thm:lb_vis}, can be derived using the invisible reader lower bound.
    
    The main idea behind Lemma \ref{lem:lb_main} is to look at the decision tree of the invisible reader. The internal nodes of the decision tree are labeled by physical register and the leaves are labeled by return values. We keep minimizing the decision tree until we find a leaf with value $v$ with a parent such that there exists a configuration $C$ where the invisible reader is at the parent (i.e. it's just about to read the register at the parent) and for it to be \textquotesingle{}unsafe\textquotesingle{} for the invisible reader to return $v$. It is \textquotesingle{}safe\textquotesingle{} to return a value at a configuration if the reader can do so without violating the semantics of regular registers. It is \textquotesingle{}unsafe\textquotesingle{} otherwise. After configuration $C$, if we never write the value $v$ again then it will forever be unsafe for the invisible reader to return the value $v$. This means that the register at the parent node can never again point to the leaf with value $v$ because if it did, the invisible reader paused at parent might execute and return $v$, an unsafe value. So the register at parent can take on one less value. This register could also appear in other parts of the decision tree, and it would have a reduced value set everywhere it appears. Therefore by removing the value $v$, we can reduce the total fanout of the implementation by 1. A more detailed version of this argument appears in the proof.
    
    Before diving into the main technical lemma, we first define some useful notation. Note that the definition doesn't care how many readers there are as long as there is at least one invisible reader.
    
    \begin{definition}
    \label{def:e_pred}
        The predicate $E(m, n, S)$ says that there exists an obstruction-free implementation of an $m$-valued regular SW register using $n$ atomic SW physical registers with total fanout $S$ such that at least one reader is invisible.
    \end{definition}
    
    \begin{lemma}
    \label{lem:lb_main}
        For $m \geq 2$, $E(m, n, S)$ implies $E(m-1, n, S-1)$.
    \end{lemma}
    
    \begin{proof}
        Suppose $E(m, n, S)$ is true. Then there exists an algorithm $A_{m}$ which satisfies the conditions from Definition $\ref{def:e_pred}$. Our goal is to construct an algorithm $A_{m-1}$ to show that $E(m-1, n, S-1)$ is also true.
        We begin by setting $A'_{m} = A_{m}$ and running the following process on $A'_{m}$. The goal of this process is to minimize $A'_{m}$ until we find a decision tree node and a configuration with desirable properties. We say that it is \textquotesingle{}safe\textquotesingle{} for a reader to return a value at a configuration $C$ if the reader can return the value without violating regular register semantics.
        
        \begin{enumerate}
            \item Let $T$ be the decision tree of an invisible reader in algorithm $A'_{m}$. Let $r$ be the reader process that runs this decision tree.
            \item Consider the set of leaves in $T$ that are closest to the root. Let $\ell$ be any leaf in this set and let $v$ be the value of $\ell$.
            \item Since $m \geq 2$, $\ell$ can't be the root of $T$, so $\ell$ must have some parent node $p$.
            \item If it is safe for $r$ to return $v$ in all configurations where $r$ is at node $p$, then replace the subtree rooted at $p$ with the leaf $\ell$ (this replacement maintains the correctness of the decision tree). Repeat from step 1 using this new algorithm $A'_{m}$. 
            \item Otherwise, we know that there exists a configuration $C$ where reader $r$ is at node $p$ and it is not safe for $r$ to return $v$. We have found the decision tree node $p$ and the configuration $C$ that we were looking for, so the process terminates.
        \end{enumerate}
        
        This process is guaranteed to terminate within a finite number of iterations because $A'_{m}$ is initially obstruction free. This means that there's a finite number of nodes between the root of $T$ and its closest leaf in the initial iteration. Each iteration reduces this distance by 1, so the process will eventually terminate.
        
        Before we get to the main part of the proof, we will take a break and fix a minor technical issue with $A'_{m}$. In step 4 of the process we may have deleted some registers and reduced the total fanout of $A'_{m}$. Ideally we would like $A'_{m}$ to have the same register count and total fanout as the original $A_{m}$. This can be achieved by \textquotesingle{}padding\textquotesingle{} $A'_{m}$ with dummy registers until it reaches $n$ register and $S$ total fanout. These registers do not impact the algorithm, they are just there to increase the space complexity and total fanout. In general if an algorithm uses $x$ registers and has $y$ total fanout, we can pad the algorithm so that it uses $x' > x$ registers and has $y' > y$ total fanout as long as $y' - y \geq x' - x$ (since adding a register increases the total fanout by at least 1). Now we can say that $A'_{m}$ uses $n$ registers and has $S$ total fanout.
        
        Everything is in place for our main argument. If there is write in progress at configuration $C$, then run it to completion and call the resulting configuration $C'$. Otherwise, there is no pending write, so we let $C'$ equal $C$. $C'$ will be the initial state of our $(m-1)$-value register implementation. 
        
        If it is not safe for reader $r$ to return value $v$ at configuration $C$, then we know that there is no partial write of $v$ at configuration $C$. Therefore it is also not safe for $r$ to return $v$ at configuration $C'$. We will keep the reader $r$ paused at node $p$. Suppose there are no more writes of $v$ after configuration $C'$.
        Then, after configuration $C'$, $r$ will never be allowed to return $v$ (if it did, it would violate regular register semantics). This means that the node $p$ will never be allowed to point to the leaf $\ell$ after configuration $C'$ (if it did, then we would resume $r$ and $r$ would read $p$ and return $v$). Note that we do not actually need to pause the reader $r$ at node $p$. Since the reader $r$ is invisible, the other processes do not know whether or not the reader is paused, so $p$ cannot be changed to point to $\ell$. Therefore if we remove $v$ from the value set starting from configuration $C'$, the algorithm $A'_{m}$ actually implements an $m-1$ valued regular register using $n$ space and $S-1$ fanout (the fanout of register $p$ is reduced by 1 since it never again points to $\ell$). This algorithm is also obstruction free and has an invisible reader $r$, so it proves that $E(m-1, n, S-1)$ is true.
    \end{proof}
    
    The next lemma is proven by inductively applying the previous lemma and it intuitively say that $S$ must be large if $m$ is large. The $S-n+1$ term in the lemma statement looks mysterious at first, but it is actually just the number of leaves in a rooted tree with $n$ internal nodes and total fan-out $S$. In the rooted tree context, total fan-out just means the sum of the number of children at each internal node.
    
    \begin{lemma}
    \label{lem:inv}
        If $E(m, n, S)$ is true, then $S - n + 1\geq m$.
    \end{lemma}
    
    \begin{proof}
        This proof is by induction on $m$. In the base case where $m=1$, this lemma holds because the total fanout $S$ is always at least as large as the number of registers $n$. This means that $S-n+1 \geq 1 = m$. Now suppose that the lemma holds for some $m-1 \geq 1$. We want to show that it holds for $m$ as well. Pick any $n$ and $S$ such that $E(m, n, S)$ is true. Since $m \geq 2$, by Lemma \ref{lem:lb_main}, we know that $E(m-1, n, S-1)$ is true as well. By the inductive hypothesis, we know that $(S-1)-n+1 \geq (m-1)$, which means that $S - n + 1\geq m$ as required.
    \end{proof}
    
    \begin{theorem}
        \label{thm:lb_inv}
        Any obstruction-free implementation of an $m$-value regular SW register using $b$-value atomic physical registers where \textbf{some reader is invisible} requires $\lceil \frac{m-1}{b-1} \rceil$ space.
    \end{theorem}
    
    \begin{proof}
        Let algorithm $A$ be such an implementation and let $n$ be the number of physical registers it uses. Algorithm $A$ shows that $E(m, n, nb)$ is true. Therefore by Lemma \ref{lem:inv}, we have that $nb - b \geq m-1$ which implies that $n \geq \frac{m-1}{b-1}$. Since $n$ must be an integer, we get a final lower bound of $n \geq \lceil \frac{m-1}{b-1} \rceil$.
    \end{proof}
    
    \begin{theorem}
        \label{thm:lb_vis}
        Any obstruction-free implementation of an $m$-value regular SW register using $b$-value atomic physical registers requires $\lceil \min(\frac{m-1}{b-1}, r+\frac{\log{m}}{\log{b}}) \rceil$ space, where $r$ is the number of readers.
    \end{theorem}
    
    \begin{proof}
        Let $A$ be such an implementation. If $A$ has an invisible reader then by Theorem \ref{thm:lb_inv}, the space complexity of $A$ is at least $\lceil \frac{m-1}{b-1} \rceil$. If all readers in $A$ are visible then there is at least one physical register for each reader. The writer requires $\lceil \frac{\log{m}}{\log{b}} \rceil$ additional physical registers to represent a value between 1 and $m$. So the total space usage is at least $\lceil r+\frac{\log{m}}{\log{b}} \rceil$. Putting these two cases together yields the desired lower bound.
    \end{proof}
    

\section{Acknowledgements}

Special thanks to Faith Ellen and Peter (Tian Ze) Chen for the many helpful discussions. I would also like to thank Alexander Spiegelman for noticing that this proof works for more than just wait-free algorithms.

	\bibliographystyle{plainurl}
	\bibliography{biblio}

\begin{thebibliography}{10}

\bibitem{attiya2004distributed}
Hagit Attiya and Jennifer Welch.
\newblock {\em Distributed computing: fundamentals, simulations, and advanced
  topics}, volume~19.
\newblock John Wiley \& Sons, 2004.

\bibitem{berger2018integrated}
Alon Berger, Idit Keidar, and Alexander Spiegelman.
\newblock Integrated bounds for disintegrated storage.
\newblock {\em arXiv preprint arXiv:1805.06265}, 2018.

\bibitem{chaudhuri2000one}
Soma Chaudhuri, Martha~J Kosa, and Jennifer~L Welch.
\newblock One-write algorithms for multivalued regular and atomic registers.
\newblock {\em Acta Informatica}, 37(3):161--192, 2000.

\bibitem{chaudhuri1994bounds}
Soma Chaudhuri and Jennifer~L Welch.
\newblock Bounds on the costs of multivalued register implementations.
\newblock {\em SIAM Journal on Computing}, 23(2):335--354, 1994.

\bibitem{chen2017step}
Tian~Ze Chen and Yuanhao Wei.
\newblock Step optimal implementations of large single-writer registers.
\newblock In {\em LIPIcs-Leibniz International Proceedings in Informatics},
  volume~70. Schloss Dagstuhl-Leibniz-Zentrum fuer Informatik, 2017.

\bibitem{HW}
M.~Herlihy and J.~Wing.
\newblock Linearizability: A correctness condition for concurrent objects.
\newblock {\em ACM Trans. Program. Lang. Syst.}, 12(3):463--492, 1990.

\bibitem{lamport1986interprocess}
Leslie Lamport.
\newblock On interprocess communication.
\newblock {\em Distributed computing}, 1(2):86--101, 1986.

\bibitem{Larsson}
Andreas Larsson, Anders Gidenstam, Phuong~Hoai Ha, Marina Papatriantafilou, and
  Philippas Tsigas.
\newblock Multiword atomic read/write registers on multiprocessor systems.
\newblock {\em {ACM} Journal of Experimental Algorithmics}, 13, 2008.

\bibitem{Peterson}
G.~L. Peterson.
\newblock Concurrent reading while writing.
\newblock {\em ACM Trans. Program. Lang. Syst.}, 5(1):46--55, 1983.

\bibitem{vidyasankar1988converting}
K~Vidyasankar.
\newblock Converting lamport's regular register to atomic register.
\newblock {\em Information Processing Letters}, 28(6):287--290, 1988.

\end{thebibliography}
\end{document}